\documentclass[10pt,twocolumn,twoside]{ieeeconf}
\IEEEoverridecommandlockouts    
\usepackage{amsthm}
\usepackage{amsmath}
\usepackage{amssymb}
\usepackage{mathtools}
\usepackage{dsfont}
\let\labelindent\relax
\usepackage{enumitem}
\usepackage{color}
\usepackage{graphicx}
\usepackage{subcaption}

\newtheorem{corollary}{\normalfont \textbf{Corollary}} 
\newtheorem{assumption}{\normalfont \textbf{Assumption}}
\newtheorem{lemma}{\normalfont \textbf{Lemma}}
\newtheorem{theorem}{\normalfont \textbf{Theorem}}
\newtheorem{proposition}{\normalfont \textbf{Proposition}}
\newtheorem{definition}{\normalfont \textbf{Definition}}
\newtheorem{remark}{\normalfont \textbf{Remark}}

\usepackage{caption}
\usepackage{subcaption}
\usepackage{cleveref}

\title{\textbf{Global and Distributed Reproduction Numbers of a Multilayer SIR Model with an Infrastructure Network}}
\author{Jos\'e I. Caiza, Junjie Qin, and Philip E. Par\'e \thanks{*Jos\'e I. Caiza, Junjie Qin, and Philip E.~Par\'e are with the Elmore Family School of Electrical and Computer Engineering at Purdue University. Emails:~jcaiza@purdue.edu, jq@purdue.edu, philpare@purdue.edu. This work was supported in part 
   by the National Science Foundation, grants
   NSF-ECCS \#2032258 and \#2238388.}}

\begin{document}

\maketitle

\begin{abstract}
   In this paper, we propose an SIR spread model in a population network coupled with an infrastructure network that has a pathogen spreading in it. 
   We develop a threshold condition to characterize the monotonicity and peak time of a weighted average of the infection states in terms of the global (network-wide) effective reproduction number. We further define the distributed reproduction numbers (DRNs) of each node in the multilayer network which are used to provide local threshold conditions for the dynamical behavior of each entity. Furthermore, we leverage the DRNs to predict the global behavior based on the node-level assumptions. We use both analytical and simulation results to illustrate that 
   the DRNs allow a more accurate analysis of the networked spreading process than the global effective reproduction number.


\end{abstract}

\section{Introduction} 

The spread of contagious diseases can be catastrophic, causing worldwide impact in a broad variety of aspects, from human losses to financial crises. The Covid-19 pandemic revealed the importance of understanding the behavior of disease-spreading to predict future outbreaks, and as a result, possibly mitigation algorithms. 
Developing these tools is pivotal to assist public health officials and politicians in their decision-making processes during epidemic outbreaks.

Networked SIR models have been studied intensively in the controls community in recent years \cite{hethcote,PARE2020345,MEI2017116}. However, a vast majority of such models only account for person-to-person interaction as the propagation mechanism. Nevertheless, diseases can spread through other means, such as water distribution systems~\cite{RANDAZZO2020115942,VERMEULEN2015109} or contaminated surfaces (e.g., public transportation services~\cite{brooks}, and hospitals~\cite{CICCOLINI2013380}). Consequently, new models have been recently proposed where a water compartment is coupled with traditional epidemic models. In~\cite{Tien2010-ch}, a Susceptible-Infected-Water-Removed (SIWR) compartmental model was proposed by coupling the classical SIR person-to-person infection with a contaminated water compartment. In~\cite{SHUAI2011118}, a cholera model was proposed and the virus propagates through both direct and indirect transmission pathways with a water compartment. In~\cite{liu}, the model proposed in~\cite{Tien2010-ch} is extended by coupling a human-contact SIS network with a single water compartment that may be contaminated. Furthermore, a generalization of the model proposed in~\cite{liu} is given in~\cite{pare_2} where a multilayer SIS model is coupled with a generic infrastructure network
. In addition to layered spread networks, there have also been multilayer models that explore the coupling between networked virus spreading and opinion dynamics~\cite{baike_peak_infection,baike_antagonistic}. To the best of our knowledge, this is the first multilayer networked model comprised of a \emph{networked SIR model and an infrastructure network}.

Reproduction numbers have been widely used to inform the public about the severity of a virus, to predict epidemic outbreaks, and to provide policymakers with information to help design mitigation strategies. More precisely, there are two reproduction numbers constantly analyzed in the literature of mathematical epidemiology. First, the \textit{basic reproduction number} of a group of individuals quantifies the expected number of infected individuals assuming the whole population is always susceptible. On the other hand, the \textit{effective reproduction number} quantifies the expected number of infected individuals considering the evolution of the susceptible proportion of the population~\cite{VANDENDRIESSCHE2017288}. The information used to compute both reproduction numbers is leveraged to provide threshold conditions around one that predict the transient and steady-state behaviors of the spreading process~\cite{PARE2020345,MEI2017116,cao,pappas}. The concept of reproduction numbers has been extended to novel networked models, e.g., networked bi-virus models~\cite{liu_bivirus} and multilayer SIS networked spreading process~\cite{pare_2,GRACY202319}.

Nonetheless, a global (network-level) analysis does not provide an accurate understanding of the local (node-level) spreading behavior, given the heterogeneity in both spreading parameters and the network structure. In~\cite{baike}, the novel concept of distributed reproduction numbers (DRNs) is introduced, which quantifies the expected number of infections resulting from the pair-wise interaction between nodes for standard networked SIS/SIR models. Further, the DRNs are used to define the reproduction number of each community in the network. In this work, we develop a DRNs framework for the novel multilayer networked SIR model \textit{with an infrastructure network} and provide threshold conditions to predict local and global (network-wide) behavior based on node-level assumptions.


In summary, the contributions of our work are:
\begin{itemize}
    \item We derive and analyze the transient behavior of a novel multilayer SIR model with an infrastructure network based on the global effective reproduction number.
    \item We define the DRNs for each node in the multilayer network, whether it is a node in the population or infrastructure network, and provide sufficient conditions to predict the monotonicity of the epidemic spreading at the node level.
    \item We leverage the DRNs to analyze the transient behavior of the overall networked spreading process.
\end{itemize}

The rest of the paper is outlined as follows. The multilayer networked SIR model with infrastructure network is developed in Section~\ref{sec:model_problem}, and the research problems of interest are stated. The analysis of equilibria and the transient behavior based on the global effective reproduction number  is given in Section~\ref{sec:network_level_analysis}. In Section~\ref{sec:DRNs}, the DRNs for any node in the multilayer network are defined and sufficient conditions are provided to predict the local spreading behavior (i.e., at the node level). Simulations illustrating our theoretical findings are provided in Section~\ref{sec:simulations}. Finally, Section~\ref{sec:conclusions} concludes this work and states future directions.

\textit{Notation}: 
The $i$th entry of a vector $x$ is denoted $x_i$. We use $\mathbf{0}$ and $\mathbf{1}$ to denote the vectors, whose entries are equal to $0$ and $1$, respectively, and use $I_n$ to denote an $n\times n$ identity matrix. For any vector $x\in\mathbb{R}^n$, we use $\text{diag}(x)$ to denote the $n\times n$ diagonal matrix whose $i$th diagonal entry equals $x_i$. For a square matrix $M$, we use $\text{diag}(M)$ to zero out the off-diagonal elements of $M$ and $[M]_j$ to denote the element in the $j$th diagonal entry. For any two real vectors $a,b\in\mathbb{R}^n$, we write $a>b$ if $a_i\geq b_i$ and $a\neq b$, and $a\gg b$ if $a_i>b_i$ for all $i\in\{1,\dots,n\}$. For a square matrix $M$, we use $\sigma(M)$ to denote the spectrum of $M$, $\rho(M)$ to denote the spectral radius of $M$, 
and $\lambda_{\max}(M) = \text{arg}\max\{\text{Re}(\lambda):\lambda \in \sigma(M)\}$.

\section{Model and Problem Formulation}~\label{sec:model_problem}

In this section, we develop a continuous-time standard networked SIR model coupled with an infrastructure network. This model will hereafter be referred to as the \textit{multilayer networked model}. We also formulate the problem we aim to analyze in this work. 

Consider a pathogen spreading over a two-layer network, namely the population or human-contact network $\mathcal{P}$ and the infrastructure network $\mathcal{I}$. 
The set of $n$ group of individuals in the population network is denoted by $\mathcal{V^P}$, i.e., $|\mathcal{V^P}|=n$, and the set of $m$ infrastructure resources is denoted by $\mathcal{V^I}$, i.e., $|\mathcal{V^I}|=m$. The set of all the nodes in the multilayer network is denoted by $\mathcal{V}$, where, $\mathcal{V}=\mathcal{V^P}\cup\mathcal{V^I}$, and $|\mathcal{V}|=|\mathcal{V^P}|+|\mathcal{V^I}|$. We allow any node to be contaminated as a consequence of its interactions with other infected groups of individuals and/or as a consequence of its interactions with the infrastructure resources. 

We denote by $s_i(t),~x_i(t),~r_i(t)$, the proportion of susceptible, infected, and recovered individuals, respectively, in a group $i$ at time $t\geq 0$. We assume that the total number of individuals in each group $i$ remains constant, i.e., $s_i(t)+x_i(t)+r_i(t)=1$. Each group of individuals $i\in\mathcal{V^P}$ has a healing rate $\gamma_i$, a person-to-person infection rate $\beta_{ij}$ for all $j\in\mathcal{V^P}$, and a person-to-resource infection rate $\beta_{ij}^w$ for all $j\in\mathcal{V^I}$. The evolution of the proportion of infected, susceptible, and recovered individuals in each group $i \in \mathcal{V^P}$ can be described as follows:
\begin{align}\label{eq:model}
    \Dot{s}_i(t) &= -s_i(t)\Bigg( \sum_{j=1}^n \beta_{ij}x_j(t) + \sum_{j=1}^m \beta_{ij}^w w_j(t)\Bigg), \nonumber \\
    \Dot{x}_i(t) &= s_i(t)\Bigg( \sum_{j=1}^n \beta_{ij}x_j(t) + \sum_{j=1}^m \beta_{ij}^w w_j(t)\Bigg) - \gamma_i x_i(t), \nonumber \\
    \Dot{r}_i(t) &= \gamma_i x_i(t), 
\end{align}
where $w_j(t)$ denotes the virus concentration in resource node $j \in \mathcal{V^I}$. The contamination of resource node $w_j(t)$ evolves as

\vspace{-5ex}

\footnotesize
\begin{align}
    \Dot{w}_j(t) &= -\gamma_j^w w_j(t) + \sum_{k=1}^m \alpha_{kj} w_k(t) - w_j(t) \sum_{k=1}^m \alpha_{jk} + \sum_{k=1}^n c_{kj}^w x_k(t), \label{eq:w_node}
\end{align}

\vspace{-2ex}

\normalsize

\noindent
where $\gamma_j^w$ denotes the decay rate of the contamination of resource node $j\in\mathcal{V^I}$, $\alpha_{kj}$ denotes the flow of the pathogen from any resource $k\in\mathcal{V^I}$, and $c_{kj}^w$ denotes the person-to-resource infection rate for all $k\in\mathcal{V^P}$.

Note that the third equation of \eqref{eq:model} is redundant given the constraint $s_i(t)+x_i(t)+r_i(t)=1$. Therefore, the model \eqref{eq:model} and \eqref{eq:w_node} in vector form becomes
\begin{subequations}\label{eq:vector_model}
    \begin{align}
        \dot{s}(t) &= -\text{diag}\big(s(t)\big)\big(Bx(t) + B_ww(t)\big), \label{eq:s_vector}\\
        \dot{x}(t) &= \text{diag}\big(s(t)\big)\big(Bx(t) + B_ww(t)\big) - Dx(t), \label{eq:x_vector} \\
        \dot{w}(t) &= -D_ww(t) + A_ww(t) + C_wx(t), \label{eq:w_vector}
    \end{align}
\end{subequations}
where $B = [\beta_{ij}]_{n\times n}$,~$B_w = [\beta_{ij}^w]_{n\times m}$, $D$ and $D_w$ are diagonal matrices with the healing rates $\gamma_i$ and $\gamma_i^w$, respectively,~$A_w$ has negative diagonal entries equal to $\alpha_{jj}-\sum_k\alpha_{kj}$ and off-diagonal entries equal to $\alpha_{kj}$, and $C_w=[c_{jk}^w]_{m\times n}$. Therefore, the columns of $A_w$ sum to zero, i.e., $A_w^\top\mathbf{1}=0$.

System \eqref{eq:vector_model} can be written more compactly using:
\begin{align}\label{eq:matrices}
    z(t)&\coloneqq \begin{bmatrix}
        x(t) \\ w(t)
    \end{bmatrix},~~~H\big(s(t)\big)\coloneqq \begin{bmatrix}
        \text{diag}\big(s(t)\big) & 0 \\ 0 & I_m
    \end{bmatrix}, \nonumber\\
    B_f &\coloneqq \begin{bmatrix}
        B & B_w \\ C_w & A_w - \text{diag}(A_w)
    \end{bmatrix}, \nonumber\\
    D_f &\coloneqq \begin{bmatrix}
        D & 0 \\ 0 & D_w - \text{diag}(A_w)
    \end{bmatrix}.
\end{align}

Therefore, \eqref{eq:vector_model} can be written as 
\begin{align}\label{eq:z}
    \dot{z}(t) &= \big(H(s(t))B_f - D_f\big)z(t).
\end{align}

We impose the following assumptions on the system parameters.

\begin{assumption}\label{assump:parameters}
    For all $i,j \in \mathcal{V^P},~\gamma_i>0,~\beta_{ij}\geq0$. For all $i\in\mathcal{V^P},~j\in\mathcal{V^I},~\gamma_j^w- \alpha_{jj} +\sum_{k} \alpha_{kj} > 0,~\beta_{ij}^w\geq 0,~c_{ij}^w\geq0$, with at least one $i$ such that $c_{ij}^w>0$,  and one $j$ such that $\beta_{ij}^w>0$. Moreover, $B$ and $A_w$ are irreducible.
\end{assumption}

The assumption imposed on $\beta_{ij}^w$ and $c_{ij}^w$ ensures the coupling between the population network and the resource nodes, and therefore $B_f$ is irreducible. Moreover, Assumption~\ref{assump:parameters} ensures matrix $D_f$ is invertible since both the healing rates for the population nodes $\gamma_i$, and the contamination decay of the resource nodes $\gamma_j^w- \alpha_{jj} +\sum_{k} \alpha_{kj}$ is positive.
We first show that the system is well-defined. That is, since $s_i(t),~x_i(t),~r_i(t)$ represent proportions of a given population $i\in\mathcal{V^P}$, they must not exceed one or go negative.  Moreover, the concentration of the virus $w_j(t)$ at each resource node $j\in\mathcal{V^I}$ must never be negative. Otherwise, these states lack physical meaning.
\begin{lemma}\label{lem:possitivity}
    Suppose Assumption~\ref{assump:parameters} holds, $s_i(0),~x_i(0),$ $s_i(0)+x_i(0) \in [0,1]$ for all $i\in\mathcal{V^P}$ and $w_j(0)\geq0$ for all $j\in\mathcal{V^I}$. Then, $x_i(t) \in [0,1]$ for all $i\in[0,1]$ and $w_j(t) \geq 0$ for all $j\in\mathcal{V^I}$, for all $t\geq0$.
\end{lemma}
\begin{proof}
    Suppose that at some time $\tau$, $s_i(\tau)$, $x_i(\tau)$, $s_i(\tau)+x_i(\tau)\in[0,1]$ for all $i\in\mathcal{V^P}$ and $w_j(\tau)\geq 0$ for all $j\in\mathcal{V^I}$. First consider any resource node $j\in\mathcal{V^I}$. If $w_j(\tau)=0$, then from Assumption~\eqref{assump:parameters} and \eqref{eq:w_node}, $\dot{w}_j(\tau)\geq0$. Therefore, $w_j(t)\geq0$ for all $t\geq\tau$.
    
    Now consider any node $i\in\mathcal{V^P}$. If at time $\tau$ node $i$ is fully recovered, i.e., $r_i(\tau)=1,~s_i(\tau)=x_i(\tau)=0$, then from Assumption~\ref{assump:parameters} and \eqref{eq:model}, $\dot{s}_i(\tau),\dot{x}_i(\tau),\dot{r}_i(\tau)=0$. In other words, when all individuals in group $i$ have recovered at time $\tau$, the system reaches an equilibrium. Thus, $x_i(t)=0$ for all $t\geq \tau$. If at time $\tau$ the whole population at node $i$ is susceptible, i.e., $s_i(\tau)=1,~x_i(\tau)=r_i(\tau)=0$, then from Assumption~\ref{assump:parameters} and \eqref{eq:model}, $\Dot{s}_i(\tau)\leq 0,~\dot{x}_i(\tau)\geq 0$. Therefore, $x_i(t)\geq 0$ for all $t\geq \tau$. In this case, when all the population is susceptible at time $\tau$, the susceptible proportion will decrease as the infected will start increasing. Finally, consider the case when the whole population is infected, i.e., $x_i(\tau)=1,~s_i(\tau)=r_i(\tau)=0$, then from Assumption~\ref{assump:parameters} and \eqref{eq:model}, $\dot{x}_i(\tau)<0,~\dot{r}_i(\tau)>0$. Therefore, $x_i(t)\leq1$ for all $t\geq\tau$. That is, if the whole population at node $i$ becomes infected, it will start to recover and the infected proportion will decrease. 

    Since by assumption $s_i(0),~x_i(0),~s_i(0)+x_i(0)\in[0,1]$ for all $i\in\mathcal{V^P}$ and $w_j(0)\geq0$ for all $j\in\mathcal{V^I}$, the lemma follows by setting $\tau=0$. 
    \hfill 
\end{proof}

Given that we are proposing a new model and we look towards analyzing the local (node-level) dynamical behavior, in the sections that follow, we aim to provide insights related to the following problems:
\begin{enumerate}
    \item Define the global (network-wide) effective reproduction number to characterize the dynamical behavior of \eqref{eq:z}. 
    \item Define node-level effective reproduction numbers based on the DRNs to predict the transient behavior of any node $i\in\mathcal{V}$ according to threshold conditions;
    \item Provide sufficient conditions to predict the network-level behavior based on assumptions imposed on the node-level effective reproduction numbers.
\end{enumerate}

\section{Network-level Analysis}\label{sec:network_level_analysis}

This section examines the equilibrium of the multilayer model in \eqref{eq:z} and motivates the use of the global effective reproduction number to analyze the networked spreading behavior.

\subsection{Equilibria}

First, we present a result related to the monotonicity of the susceptible proportion of each population $i\in\mathcal{V^P}$.

\begin{lemma}\label{lem:s_decrease}
    If $s_i(0),x_i(0) \in [0,1]^n$, for all $i\in\mathcal{V^P}$, and $w_j(0)\geq 0$, for all $j\in\mathcal{V^I}$, the susceptible state $s_i(t)$ is monotonically decreasing, for all $i\in\mathcal{V^P}$, $t\geq0$. 
\end{lemma}
\begin{proof}
    From Assumption~\ref{assump:parameters}, we have that $\text{diag}\big(s(t)\big)\big(Bx(t) + B_ww(t)\big)$ is non-negative. Thus, by 
    \eqref{eq:s_vector}, $\dot{s}(t)$ is always non-positive. Therefore, $s(t)$ is monotonically decreasing.
    \hfill 
\end{proof}

The main takeaway from Lemma~\ref{lem:s_decrease} is that the coupling of the population with the infrastructure network does not change the SIR-like property of the susceptible proportions, i.e., $s_i(t)$ decreases with time. It is well known that the standard networked SIR model, i.e., $B_w = 0$ in \eqref{eq:s_vector} and \eqref{eq:x_vector}, has an infinite number of healthy equilibria $(s^*,\mathbf{0},\mathbf{0})$, where $r^*=\mathbf{1}-s^*$. Also, the system will never reach an endemic equilibrium. Thus, we analyze if the networked SIR model has the same steady-state behavior when the population network is coupled with an infrastructure network as in \eqref{eq:vector_model}.

\begin{proposition} \label{prop:equilibria}
    Consider the networked model in \eqref{eq:vector_model}. Let Assumption 1 hold and assume $\gamma^w_j>0$ for all $ j \in\mathcal{V^I}$. The set of equilibria has the form $(s^*,\mathbf{0},\mathbf{0})$, where $s^*\in[0,1]^n$. 
\end{proposition}
\begin{proof}
    The point $(s^*,x^*,w^*)$ is an equilibrium of \eqref{eq:vector_model} if and only if:
    \begin{align*}
        \mathbf{0} &= -\text{diag}(s^*)\big(Bx^* + B_ww^*\big) \\
        \mathbf{0} &= \text{diag}(s^*)\big(Bx^* + B_ww^*\big) - Dx^* \\
        \mathbf{0} &= -D_ww^* + A_ww^* + C_wx^*. 
    \end{align*}
    Therefore, each point of the form $(s^*,\mathbf{0},\mathbf{0})$, i.e., at a healthy state, is an equilibrium. 
    
    Conversely, adding the first two equations yields $\mathbf{0} = Dx^*$. Given that $D$ is a positive diagonal matrix, then $x^*$ must be $\mathbf{0}$. Therefore, from the last equation, we have
    \begin{equation}
    \label{eq:Aw-Dw_0}
        -D_ww^* + A_ww^* = (A_w-D_w)w^* = \mathbf{0}.
    \end{equation}
    Note that by the structure of $A_w$, 
    the column sums of $A_w$ are zero. Therefore, by Assumption 1 and since we assume $\gamma^w_j>0$ $\forall j \in\mathcal{V^I}$, the Gersgorin Disc Theorem gives that all the eigenvalues of $A_w-D_w$ are in the open left half plane and thus, there is no zero eigenvalue. Therefore, the only solution to \eqref{eq:Aw-Dw_0} is $w^*=\mathbf{0}$. 
    \hfill 
\end{proof}

The proposition implies that the system has an infinite number of healthy equilibria, whereas an endemic equilibrium does not exist, i.e., $\nexists~x^*, w^*$ s.t. $x^* \neq 0, w^*\neq 0$. Therefore, the threshold behavior that is determined by the basic reproduction number for the networked SIS model, namely $R_0=\rho(D^{-1}B)$~\cite{PARE2020345}, does not apply to our model. 
Rather, we are interested in understanding the transient behavior of the networked SIR model in \eqref{eq:z}. 
To that end, we turn our focus to studying the effective reproduction number of the network which is key for predicting the network-wide dynamical behavior. 

\subsection{Effective Reproduction Number}

We first introduce the concept of the global effective reproduction number for the multilayer network. 

\begin{definition}[Global Effective Reproduction Number]\label{def:R_t}
Under Assumption~\ref{assump:parameters}, the global (network-wide) effective reproduction number is denoted by $R(t) = \rho\big(H(s(t))D_f^{-1}B_f\big)$, where $H(s(t))$, $B_f$, and $C_f$ are given in \eqref{eq:matrices}.
\end{definition}

Given that $R(t)$ accounts for the evolution of the virus for the multilayer networked SIR model, we characterize the dynamical behavior of a weighted average of the components of $z(t)$. To that end, we first define the notion of peak infection time.  

\begin{definition}[Peak Infection Time]\label{def:peak_infection_time}
    Let $v(t)\in\mathbb{R}^{n+m}$ be a positive normalized vector. A peak infection time $\tau_p$ of the weighted average $v(\tau_p)^\top z(t)$ is such that $v(\tau_p)^\top z(\tau_p)>v(\tau_p)^\top z(t)$ for all $t\geq0$, $t\neq\tau_p$.
\end{definition}

We now characterize the threshold behavior of the multilayer networked model~\eqref{eq:vector_model} in the following theorem.

\begin{theorem}\label{theo:global_effective}
    Assume Assumption~\ref{assump:parameters} holds and $s(0)\gg\mathbf{0}$ and $z(0)>\mathbf{0}$. Let $v(t)$ be the normalized left eigenvector associated with the eigenvalue $\lambda_{\max}(t)$ of the matrix $H(s(t))B_f-D_f$. The following claims hold:
    \begin{enumerate}[label=\roman*)]
        \item The effective reproduction number $R(t)$ is monotonically decreasing with respect to $t$. \label{theo:R_decreasing}
        \item Assume there is a time $\tau>0$ such that $v(\tau)^\top z(t)$ is increasing for all $t\leq \tau$ in a sufficiently small time interval $[t,\tau]$. Then, $R(t)>1$, for all $t\leq\tau$. \label{theo:R_great_1}
        \item Let the time $\tau_p\geq 0$ satisfy $R(\tau_p)=1$. Then, the weighted average $v(\tau_p)^\top z(t)$ reaches a maximum value at $t=\tau_p$ and the peak infection is unique.\label{theo:peak_global}
        \item If the weighted average $v(\tau)^\top z(t)$ is decreasing for all $t\geq \tau$ in a sufficiently small time interval $[\tau,t]$. Then $R(t)<1$, for all $t\geq\tau$. \label{theo:R_less_1}
    \end{enumerate}
\end{theorem}

\begin{proof}
    Statement~\ref{theo:R_decreasing} is the immediate consequence of Lemma~\ref{lem:s_decrease} and Definition~\ref{def:R_t}, i.e., the susceptible states of each location in the population network are monotonically decreasing and $R(t)$ is a positive continuous linear function of $s(t)$.
    
    Now we consider the case when the weighted average of the components of $z(t)$ increases (statement \ref{theo:R_great_1}). Since $s(0)\gg\mathbf{0}$, by \eqref{eq:s_vector}, $s(t)\gg\mathbf{0}$ for all $t\leq\tau$. Therefore, since under Assumption~\ref{assump:parameters} $H(s(t))B_f-D_f$ is an irreducible Metzler matrix, $\lambda_{\max}(t)$ is a simple eigenvalue for all $t\in[0,\tau]$ by~\cite[Lemma~$7$]{varga}. Additionally, the normalized left eigenvector $v(t)$ satisfies $v(t)\gg\mathbf{0}$ and $v(t)^\top\mathbf{1}_{n+m}=1$ for all $t\in[0,\tau]$.
    By left multiplying $v(\tau)^\top$ on both sides of \eqref{eq:z} we have:
    \begin{align*}
        v(\tau)\dot{z}(t) &= v(\tau)^\top \big(H(s(t))B_f-D_f\big)z(t) \\
        &\approx v(t)^\top \big(H(s(t))B_f-D_f\big)z(t) \\
        &=\lambda_{\max}(t)v(t)^\top z(t).
    \end{align*}
    If $v(\tau)^\top z(t)$ is increasing for all $t\leq \tau$, then $v(\tau)^\top \dot{z}(t)>0$ for all $t\leq\tau$. Since $v(\tau)^\top \dot{z}(t)>0$ for all $z(t)>\mathbf{0}$ and $v(\tau)\gg\mathbf{0}$, then $\lambda_{\max}(t)>0$. In an analogous way, statement \ref{theo:R_less_1} can be proven.
    

    For statement \ref{theo:peak_global}, we start by recalling~\cite[Proposition~$1$]{liu_bivirus}: $R(\tau_p)=1$ is equivalent to $\lambda_{\max}(\tau_p)=0$. Thus, we have that $\lambda_{\max}(\tau_p)v(\tau_p)^\top z(\tau_p)=0$ for $z(t)>0$. By the definition of $\lambda_{\max}(\tau_p)$, we obtain 
    \begin{align*}
        0&=\lambda_{\max}(\tau_p)v(\tau_p)^\top z(\tau_p) \\
        &=v(\tau_p)^\top\big(H(s(\tau_p)B_f-D_f\big)z(\tau_p) \\
        &= \frac{d}{dt}\big(v(t)^\top z(t)\big)\Big|_{t=\tau_p }.
    \end{align*}
    Thus, when $R(\tau_p)=1$, we can guarantee there will be an extremum point at $\tau_p$. We still need to prove that it is a maximum and that it is unique. Based on statement \ref{theo:R_decreasing}, we have that $R(t)$ is decreasing, which implies that $R(\tau_p-\varepsilon)>R(\tau_p)>R(\tau_p+\varepsilon)$. Given that $R(\tau_p)=1$, we have that, for the interval $[\tau_p-\varepsilon,\tau_p)$, the weighted average $v(\tau_p)^\top z(t)$ is increasing, and, for the interval $(\tau_p,\tau_p+\varepsilon]$, it is decreasing. Therefore, at time $\tau_p$ the weighted average will have a peak infection. Given the monotonicity of $R(t)$, once the weighted average is decreasing at a time $\tau$, i.e., $R(\tau)<1$, there does not exist a time $\tau'>\tau$ such that $R(\tau')>1$. Thus, the peak infection is unique at time $\tau_p$.
    \hfill 
\end{proof}

An important insight that Theorem~\ref{theo:global_effective} provides is that $R(\tau_p)=1$ is a sufficient condition to predict the peak infection, and thus guarantees that $\tau_p$ is the peak infection time of the weighted average of the infection across the multilayer networked SIR model. Another takeaway is that the uniqueness of the peak infection is due to the monotonicity of $R(t)$ with respect to $t$ (Theorem~\ref{theo:global_effective}~{\ref{theo:R_decreasing}}). Hence, we have the following direct conclusion from Theorem~\ref{theo:global_effective} related to the dynamical behavior of the weighted average once it starts decreasing.

\begin{corollary}\label{coro:R_always_decrease}
    Assume that at a given time $\tau\geq 0,~R(\tau)<1$. Then, the weighted average $v(\tau)^\top z(t)$, for $t\geq\tau$, is monotonically exponentially decreasing to zero.
\end{corollary}
   
\begin{proof}
     We first compute the derivative of the weighted average
     \begin{align*}
         \frac{d}{dt}\big(v(\tau)^\top z(t)\big) &= v(\tau)^\top \big(H(s(t)B_f-D_f)\big)z(t)\\
         &\leq v(\tau)^\top \big(H(s(\tau)B_f-D_f)z(t)\\
         &=\lambda_{\max}(\tau)v(\tau)^\top z(t),
     \end{align*}
     where the inequality comes from the decreasing monotonicity of $s(t)$ and the equality from the definition of $v(\tau)$. Note that
     \[v(\tau)^\top z(t) \leq v(\tau)^\top z(0)e^{\lambda_{\max}(\tau)\cdot t}.\]
     From~\cite[Proposition~$1$]{liu_bivirus}, $R(\tau)<1$ is equivalent to $\lambda_{\max}(\tau)<0$. Thus, the right-hand side decays exponentially to zero.
     \hfill 
\end{proof}






So far, we have provided sufficient conditions to characterize the dynamical behavior of a weighted average of $z(t)$. However, the claims made for the weighted average at the network level cannot be extrapolated to predict local spreading behavior. Moreover, given that $x(t)$ and $w(t)$ are modeled in different ways, i.e., $x(t)$ has nonlinear dynamics and $w(t)$ is linear, these distinctions are not captured by the behavior of the weighted average. 
Therefore, in the following section, we define the distributed reproduction numbers for the novel multilayer networked SIR model and elaborate on sufficient conditions to analyze the local dynamical behavior. 


\section{Distributed Reproduction Numbers}\label{sec:DRNs}

Based on the threshold conditions for the weighted average in Theorem~\ref{theo:global_effective}, we cannot guarantee that all nodes in the network will reach the peak infection at the same time. Further, $R(0)<1$ does not necessarily imply that the epidemic is dying out immediately for all $i\in\mathcal{V}$. Thus, we introduce the concept of DRNs to provide a finer granularity in the analysis of the networked spreading process. 
 The distributed reproduction numbers for each node in the multilayer network are defined in a \textit{pairwise} fashion, considering the rate at which a particular node $i$ becomes infected by a node $j$ over time. In general, the DRNs of any node $i$ associated with node $j$ is the expected number of new infections caused by the scaled infected proportion of individuals in node $j$ given that some individuals in node $i$ may no longer be susceptible. 


The previous definition establishes that a particular node will have as many reproduction numbers as in-neighbors has. Since the infection level of a particular node is influenced by different sources, i.e., human interaction or resource nodes, it is necessary to analyze the infection contribution of different in-neighbors at the same scale. We define the DRNs for the population and infrastructure network.

\begin{definition}[Population Network DRNs]\label{def:DRN_population}

Let Assumption~\ref{assump:parameters} hold and assume $x_i(t)>0$. For each location $i \in \mathcal{V^P}$, the DRNs are given by the following piecewise function
\[R_{ij}(t) = s_i(t)\frac{I_i(j,t)}{\gamma_i x_i(t)}, 
\]
where
\[I_i(j,t)=\begin{cases}
    \beta_{ij}x_j(t), & \text{if $j\in\mathcal{V}^P$} \\
    \beta_{ij}^ww_j(t), & \text{if $j\in\mathcal{{V}}^I$}.
\end{cases}\]

\end{definition}

Note that $R_{ij}(t)$ in Definition~\ref{def:DRN_population} also accounts for the new infection cases generated within group $i \in \mathcal{V^P}$ itself, i.e., $R_{ii}(t)$. 

From Lemma~\ref{lem:s_decrease}, we know that the susceptible proportion of a given location in the population network is decreasing w.r.t. time. However, since the contamination of the resource nodes is interpreted as the flow of the virus through the node, we can argue that all $j\in\mathcal{V^I}$ are always susceptible to contamination. On the other hand, the concentration of the virus has a decay rate $\gamma_j^w$ and it is diluted among other resource nodes according to \eqref{eq:w_node}. Thus, the total healing rate for a resource node $j \in \mathcal{V^I}$ is given by $\gamma_j - \text{diag}(A_w)_j$. Consequently, we characterize the DRNs of the resource nodes in the following definition.

\begin{definition}[Infrastructure Network DRNs] \label{def:DRN_infrastructure}
    Let Assumption~\ref{assump:parameters} hold and assume $w_j(t)>0$. For each resource node $j \in \mathcal{V^I}$, the DRNs are given by the following piecewise function 
    \vspace{-2ex}
    \[ R_{jk}(t)= \frac{I^w_j(k,t)}{\big(\gamma_j^w - \text{diag}(A_w)_{jj}\big)w_j(t)}, \]
    where
    \[I^w_j(k,t)=\begin{cases}
        c_{kj}^wx_k(t), & \text{if $k\in\mathcal{V^P}$} \\ 
        \alpha_{kj}w_k(t), & \text{if $k \in \mathcal{V^I}$}.
    \end{cases}\]
\end{definition}

To leverage the information that the DRNs capture, one can compute the evolution of the infected proportion of a particular node $i$, associated with a pair of nodes $(j,k)$ where $j \in \mathcal{V^P}$ and $k \in \mathcal{V^I}$. For a node in the population network, $i \in \mathcal{V^P}$, we define its infected proportion associated with a pair of nodes $(j,k)$, where $j\in \mathcal{V^P}$ and $k\in\mathcal{V^I}$, as 
\begin{equation}\label{eq:xijk}
    \dot{x}_{ij}^k(t) = s_i(t)\big(\beta_{ij}x_j(t) + \beta_{ik}^ww_k(t)\big) - \gamma_ix_i(t). 
\end{equation}

\begin{lemma} \label{lem:x_ijk}
    The infected proportion at node $i$ associated with a pair of nodes $(j,k)$ where $j\in \mathcal{V^P}$ and $k\in\mathcal{V^I}$, denoted by $x_{ij}^k(t)$, is increasing if and only if $R_{ij}(t)+R_{ik}(t)>1$, and it is decreasing if and only if $R_{ij}(t)+R_{ik}(t)<1$. 
\end{lemma}

\begin{proof}
    The following proof holds for either $i \in \mathcal{V^P}$ or $i \in \mathcal{V^I}$. We prove the case for any location $i$ in the population network, i.e., $i\in\mathcal{V^P}$. 
    
    From \eqref{eq:xijk}, $\dot{x}_{ij}^k(t)>0$ if and only if 
    \[\frac{s_i(t)}{\gamma_ix_i(t)}\big(\beta_{ij}x_j(t) + \beta_{ik}^ww_k(t)\big)>1.\]
    By Definition~\ref{def:DRN_population}, it follows that $\dot{x}_{ij}^k(t)$ is increasing if and only if $R_{ij}(t)+R_{ik}(t)>1$. We can use the same process to show that $\dot{x}_{ij}^k(t)$  is decreasing if and only if $R_{ij}(t)+R_{ik}(t)<1$.
\hfill 
\end{proof}

Up to this point, we have defined the DRN's based on the interaction between two nodes depending on which layer of the network they belong to (Definitions~\ref{def:DRN_population} and \ref{def:DRN_infrastructure}). In addition, we have shown the threshold behavior for the rate of infection at each node in terms of the DRN's (Lemma~\ref{lem:x_ijk}). However, Lemma~\ref{lem:x_ijk} only accounts for the monotonicity of the infection at a given node considering a particular combination of two sources of infection. To characterize the spreading process for each node in the multilayer network, we further define the \textit{local effective reproduction number (LERN)} which builds off the DRNs in Definitions~\ref{def:DRN_population} and \ref{def:DRN_infrastructure}.

\begin{definition}[LERNs]\label{def:local_ERN}
    For any node $i\in\mathcal{V}$, the LERN, denoted by $R_i(t)$, is the expected number of new infections caused by all possible sources of infection/contamination 
    \[R_i(t)=\sum_{j\in\mathcal{V}}R_{ij}(t).\]
\end{definition}

We characterize a threshold dynamical behavior result for any node $i\in\mathcal{V}$ in the following theorem.

\begin{theorem}\label{theo:Rt_local}
    Let Assumption~\ref{assump:parameters} hold. For a given location $i \in \mathcal{V^P}$,
     assume $x_i(t)>0$. For a given resource $j \in \mathcal{V^I}$, assume $w_j(t)>0$. Then, for a given time $t\geq 0$, the following claims hold:
    \begin{enumerate}[label=\roman*)]
        \item 
        $R_i(t)>1$ if and only if $x_i(t)$ is increasing, and 
        $R_j(t)>1$ if and only if $w_j(t)$ is increasing. \label{theo:local_ERN_increasing}
        \item 
        $R_i(t)<1$ if and only if $x_i(t)$ is decreasing, and 
        $R_j(t)<1$ if and only if $w_j(t)$ is decreasing. \label{theo:local_ERN_decreasing}
    \end{enumerate}
\end{theorem}

\begin{proof}
    Regarding statement \ref{theo:local_ERN_increasing}, we will prove the case for an arbitrary node $i\in\mathcal{V^P}$. The same procedure holds for the case of any resource node $j \in \mathcal{V^I}$.
    
    $(\Leftarrow):$ According to Definition~\ref{def:local_ERN}, the effective reproduction number of a location is given by $R_i(t)=\sum_{j\in\mathcal{V}}R_{ij}(t)$
    . If we evaluate the summation operator, we have that $R_i(t)=\sum_{j\in\mathcal{V^P}}R_{ij}(t) + \sum_{j\in\mathcal{V^I}}R_{ij}(t)$. Based on Definition~\ref{def:DRN_population}, we can rewrite $R_i(t)$ as 
    \[R_i(t)=\frac{s_i(t)}{\gamma_ix_i(t)}\Bigg(\sum_{j\in\mathcal{V^P}}\beta_{ij}x_j(t)+\sum_{j\in\mathcal{V^I}}\beta_{ij}^ww_j(t)\Bigg).\]
    Hence, if $R_i(t)>1$, it is true that 
    \begin{align}
        \frac{s_i(t)}{\gamma_ix_i(t)}\Bigg(\sum_{j\in\mathcal{V^P}}\beta_{ij}x_j(t)+\sum_{j\in\mathcal{V^I}}\beta_{ij}^ww_j(t)\Bigg) &>1 \label{eq:Ri_pop}\\ 
        s_i(t)\Bigg(\sum_{j\in\mathcal{V^P}}\beta_{ij}x_j(t) + \sum_{j\in\mathcal{V^I}}\beta_{ij}^ww_j(t)\Bigg) - \gamma_ix_i(t) &>0, \label{eq:derivative}
    \end{align}
    which implies $\dot{x}_i(t)>0$. Therefore, $x_i(t)$ is increasing. 

    $(\Rightarrow):$ If the infected proportion at node $i$ is increasing, then $\dot{x}_i(t)>0$, which is given by \eqref{eq:derivative}.  By rearranging terms, we obtain \eqref{eq:Ri_pop}. It immediately follows that $R_i(t)>1$, given that $x_i(t)>0$.

    To prove the decreasing behavior in statement \ref{theo:local_ERN_decreasing}, we follow the same technique as for statement~\ref{theo:local_ERN_increasing}.
    \hfill 
\end{proof}

\begin{remark}
    Note that in Theorem~\ref{theo:Rt_local} all the statements hold with strict inequality. That is, up to this point, we cannot guarantee that when $R_i(t)=1$, a node $i\in\mathcal{V}$ has reached a peak infection or further conclude it is unique.
\end{remark}

\section{Global Behavior from LERNs}
In this section, we aim to integrate the framework developed in Sections~\ref{sec:network_level_analysis} and \ref{sec:DRNs}. More precisely, given some assumptions on the LERNs, we infer a threshold behavior at the network level.
In the analysis that follows, we denote $\Tilde{R}_{ij}(t)$ as the unscaled DRN of node $i$ associated with node~$j$. Based on Assumption~\ref{assump:parameters}, Definitions~\ref{def:DRN_population} and \ref{def:DRN_infrastructure}, we have that, for any location $i \in \mathcal{V^P}$, the unscaled DRNs are given by
\begin{align} \label{eq:unscaled_pop}
    \Tilde{R}_{ij}(t) &= s_i(t)\frac{\Tilde{I}_{ij}}{\gamma_i},~\text{where}~\Tilde{I}_{ij}= \begin{cases}
        \beta_{ij}, & \text{if $j\in\mathcal{V^P}$} \\
        \beta_{ij}^w, & \text{if $j\in\mathcal{V^I}$},
    \end{cases}
\end{align}
and, for any resource node $j\in\mathcal{V^I}$, the unscaled DRNs are computed as 
\begin{align} \label{eq:unscaled_infra}
    \Tilde{R}_{jk} &= \frac{\Tilde{I}_{ij}^w}{\gamma_j^w-\text{diag}(A_w)_j},~\text{where}~\Tilde{I}_{ij}= \begin{cases}
        c_{kj}^w, & \text{if $j\in\mathcal{V^P}$} \\
        \alpha_{kj}, & \text{if $j\in\mathcal{V^I}$}.
    \end{cases}
\end{align}

We leverage the unscaled DRNs to define submatrices that describe the infection contribution considering all interactions across the multilayer network, i.e., person-to-person, resource-to-person, person-to-resource, resource-to-resource.
Let the matrix $\mathcal{R_P}(t)$ denote the DRNs of each node $i\in\mathcal{V^P}$ associated with all nodes $j\in\mathcal{V^P}$. In other words, $\mathcal{R_P}(t)$ describes the infection contribution between all the locations in the population network in a distributed fashion. More precisely, the $(i,j)-\text{th}$ entry of $\mathcal{R_P}(t)$ is given by $[\mathcal{R_P}(t)]_{ij}=\Tilde{R}_{ij}(t)$ for all $i,j\in\mathcal{V^P}$. On the other hand, the matrix $\mathcal{R_{I\rightarrow P}}(t)$ accounts for the infection contribution from all the resource nodes in the infrastructure network to the population nodes, i.e., $[\mathcal{R_{I\rightarrow P}}(t)]_{ij} = \Tilde{R}_{ij}(t)$ for all $i\in\mathcal{V^P}$ and all $j\in\mathcal{V^I}$. In the same way, the contamination of the resource nodes from each location in the population network is given by $[\mathcal{R_{P\rightarrow I}}(t)]_{jk}=\Tilde{R}_{jk}(t)$ for all $j\in\mathcal{V^I}$ and $k\in\mathcal{V^P}$. Finally, $[\mathcal{R_I}(t)]_{jk}=\Tilde{R}_{jk}(t)$ for all $j,k\in \mathcal{V^I}$, accounts for the contamination between resource nodes in the infrastructure network. Now, we define the effective reproduction matrix of the multilayer network, denoted by $\mathcal{R}(t)$.

\begin{definition}[Global Effective Reproduction Matrix] \label{def:Rt_matrix}
    The effective reproduction matrix of the network is given by 
    \[
    \mathcal{R}(t) = \begin{bmatrix}
        \mathcal{R_P}(t) & \mathcal{R_{I\rightarrow P}}(t) \\
        \mathcal{R_{P\rightarrow I}}(t) & \mathcal{R_I}(t)
    \end{bmatrix}, 
    \]
    where $\mathcal{R_P}(t)\in \mathbb{R}^{n\times n}$, $\mathcal{R_{I\rightarrow P}}(t) \in \mathbb{R}^{n\times m}$, $\mathcal{R_{P\rightarrow I}}(t) \in \mathbb{R}^{m\times n}$ and $\mathcal{R_I}(t) \in \mathbb{R}^{m\times m}$. 
\end{definition}


From \eqref{eq:matrices}, \eqref{eq:unscaled_pop}, and \eqref{eq:unscaled_infra}, it is straightforward to conclude that $\mathcal{R}(t)=H(s(t))D_f^{-1}B_f$. Hence, we can obtain the global effective reproduction number by computing the spectral radius of the global effective reproduction matrix, i.e., $R(t)=\rho(\mathcal{R}(t))$. We leverage the connection between the DRNs and the structure of the global effective reproduction matrix, to predict the network-level behavior based on node-level assumptions. To that end, we present the following theorem.

\begin{theorem}\label{theo:node_to_network}
    Assume that Assumption~\ref{assump:parameters} holds and the system has not reached a healthy equilibrium, i.e., $z(t)\neq\mathbf{0}$ for a given time $t\geq0$. The following statements hold:
    \begin{enumerate}[label=\roman*)]
        \item If $R_i(t)>1$ for all $i\in\mathcal{V}$, then $\rho(\mathcal{R}(t))>1$, \label{theo:all_great_1}
        \item If $R_i(t)=1$ for all $i\in\mathcal{V}$, then $\rho(\mathcal{R}(t))=1$, \label{theo:all_1}
        \item If $R_i(t)<1$ for all $i\in\mathcal{V}$, then $\rho(\mathcal{R}(t))<1$. \label{theo:all_less_1}
    \end{enumerate}
\end{theorem}
\begin{proof}
    Note that since the system has not reached a healthy equilibrium, $s_i(t),x_i(t), w_j(t) > 0$ for all $i\in\mathcal{V^P}, j\in \mathcal{V^I}$.
    Also, note that the distributed reproduction number of a location $i\in\mathcal{V^P}$ associated with a location $j\in\mathcal{V^P}$ can be written in terms of its unscaled DRN as 
    \[R_{ij}(t)=\frac{1}{x_i(t)}\Tilde{R}_{ij}(t)x_j(t),\]
    and in the case $j\in\mathcal{V^I}$, we have that 
    \[R_{ij}(t)=\frac{1}{x_i(t)}\Tilde{R}_{ij}(t)w_j(t) .\]
    In the same fashion, we can reconstruct the DRNs of any resource node $j\in\mathcal{V^I}$ based on its unscaled DRNs. Thus, in vector form we obtain
    \begin{align} \label{eq:Ri_reconstruct}
        \begin{bmatrix}
        R_1(t) & \cdots & R_n(t) & R_{n+1}(t) & \cdots & R_{n+m}(t)
    \end{bmatrix}^\top& \nonumber\\
    & \hspace{-35ex} = \big[\text{diag}(z(t))^{-1}\mathcal{R}(t)\text{diag}(z(t))\big]\mathbf{1},
    \end{align}
    where the first $n$ coordinates are associated with the $n$ groups of individuals in the population network and the last $m$ coordinates with the resource nodes in the infrastructure network. 
    
    For statement \ref{theo:all_1}, by \eqref{eq:Ri_reconstruct}, we have that the matrix $\text{diag}(z(t))^{-1}\mathcal{R}(t)\text{diag}(z(t))$ is a row stochastic matrix. Based on the fact that the spectral radius of a row stochastic matrix is 1, we have $\rho\big(\text{diag}(z(t))^{-1}\mathcal{R}(t)\text{diag}(z(t))\big)=1$.
    Since similarity transformations preserve the characteristic equation and the eigenvalues, we can conclude that $\rho\big(\text{diag}(z(t))^{-1}\mathcal{R}(t)\text{diag}(z(t))\big)=\rho(\mathcal{R}(t))$. Therefore, if $R_i(t)=1$ for all $i\in\mathcal{V}$, then we can guarantee $\rho(\mathcal{R}(t))=1$.

    For statement \ref{theo:all_great_1}, since $R_i(t)>1$ for all $i\in\mathcal{V}$, we have that $R_i(t)=\sum_{j\in\mathcal{V}}R_{ij}(t)>1$ by definition~\ref{def:local_ERN}. Assume by way of contradiction that $\rho(\mathcal{R}(t))=$ $\rho\big(\text{diag}(z(t))^{-1}\mathcal{R}(t)\text{diag}(z(t))\big)<1$ (we do not include the equality case, i.e, $\rho(\mathcal{R}(t))=1$, since it falls in the first part of this proof). By decreasing some of the entries of $\mathcal{R}(t)$, we construct a new matrix $\Tilde{\mathcal{R}}(t)$, such that $\rho\big(\text{diag}(z(t))^{-1}\Tilde{\mathcal{R}}(t)\text{diag}(z(t))\big)=1$, i.e., a row stochastic matrix. Given that $\mathcal{R}(t)$ is a nonnegative and irreducible matrix, by Assumption~\ref{assump:parameters} and since $s_i(t) > 0$ for all $i\in\mathcal{V^P}$, the spectral radius of an irreducible matrix will increase if one of its entries increases~\cite[Theorem 2.7]{varga}. Thus, we have that 
    
    \vspace{-2ex}
    
    \small
    \begin{align*}
        \rho\big(\text{diag}(x(t))^{-1}\mathcal{R}(t)\text{diag}(x(t))&>\rho\big(\text{diag}(x(t))^{-1}\Tilde{\mathcal{R}}(t)\text{diag}(x(t))\big) \\
        &=1
    \end{align*}
    
    \normalsize
    
    \noindent
    which contradicts the assumption that $\rho(\mathcal{R}(t))\leq 1$. Therefore, we must have that $\rho\big(\text{diag}(x(t))^{-1}\mathcal{R}(t)\text{diag}(x(t))>1$. Statement \ref{theo:all_less_1} follows an analogous argument.
    \hfill 
\end{proof}    
    This proof was inspired by the proof of \cite[Theorem $2$]{baike}.

\section{Simulations}\label{sec:simulations}

For the simulations, we consider a network of $10$ population nodes and $5$ resource nodes. The parameters are picked uniformly at random from the following intervals. For all $i \in \mathcal{V^P}$, $\gamma_i\in[1,3]$. For all $j\in\mathcal{V^I}$, $\gamma_j^w\in[0.6,0.75]$. The entries of the matrix $B$ are selected from the interval $[0.01,0.338]$, the entries of the matrix $B_w$ are selected from $[0.01,0.2]$, and $C_w$ is computed as $C_w = B_w^\top - 0.01\mathbf{1}_{m\times n}$. Finally, for all $j,k\in\mathcal{V^I},~\alpha_{jk}\in[0,2]$. We assume $s(0)=0.95\mathbf{1}_n$ and $r(0)=\mathbf{0}_n$. The initial contamination level for the resource nodes $w_j(0),~j \in\mathcal{V^I}$ is picked uniformly at random from the interval $[0,1]$. All plots use the same set of parameters.

\begin{figure}
    \centering
    \begin{subfigure}{0.5\textwidth}
        \centering
        \includegraphics[width=0.75\linewidth]{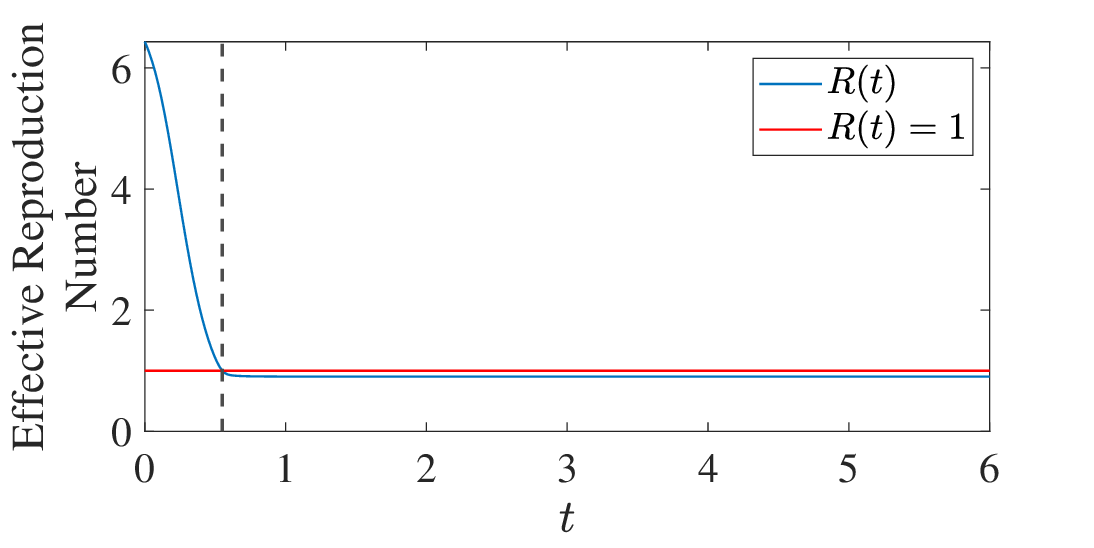}
    \end{subfigure}
    \begin{subfigure}{0.5\textwidth}
        \centering
        \includegraphics[width=0.75\linewidth]{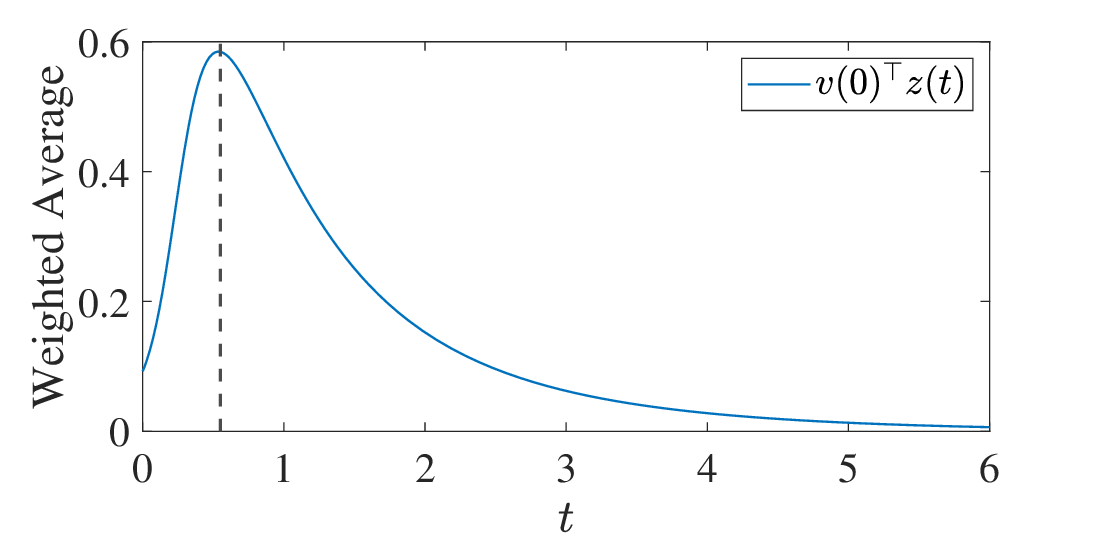}
    \end{subfigure}
    \caption{Evolution of the global effective reproduction number $R(t)$ (top) and the weighted average $v(0)^\top z(t)$ (bottom).}
    \label{fig:network-level}    
\end{figure}

Consistent with the claims in Theorem~\ref{theo:global_effective}, we see in Fig.~\ref{fig:network-level} that the global effective reproduction number monotonically decreases with time, and the peak infection time of the weighted average $v(0)^\top z(t)$ coincides with $R(t)=1$ (indicated by the vertical dashed lines in Fig.~\ref{fig:network-level}). Moreover, $v(0)^\top z(t)$ increases when $R(t)>1$ and decreases when $R(t)<1$. By inspecting the weighted average plot, we can conclude that the outbreak dies out in the population and infrastructure networks, i.e., $x(t)=\mathbf{0}$ and $w(t)=\mathbf{0}$ as $t\rightarrow \infty$, consistent with Proposition~\ref{prop:equilibria}.

Using the same parameters as in Fig.~\ref{fig:network-level}, we illustrate the evolution of the infected proportion $x_i(t)$ and the LERNs $R_i(t)$ for $i\in\{2,3,6,8,9\}$ of the population network; see Fig.~\ref{fig:node-level-pop}. Consistent with Theorem~\ref{theo:Rt_local}, for a given $i\in\mathcal{V^P}$, the infected proportion increases when $R_i(t)>1$, and decreases otherwise. Note that the dashed lines 
in Fig.~\ref{fig:node-level-pop}, corresponding to $i=\{2,3\}$ in the population network, indicate that node $i$'s peak infection time occurs when $R_i(t)=1$. Therefore, we conjecture that, for node $i\in\mathcal{V^P}$ in the population network,  $R_i(t)=1$ is a necessary and sufficient condition for identifying the peak infection time on the node level.

\begin{figure}
    \centering
    \begin{subfigure}{0.5\textwidth}
        \centering
        \includegraphics[width=0.75\linewidth]{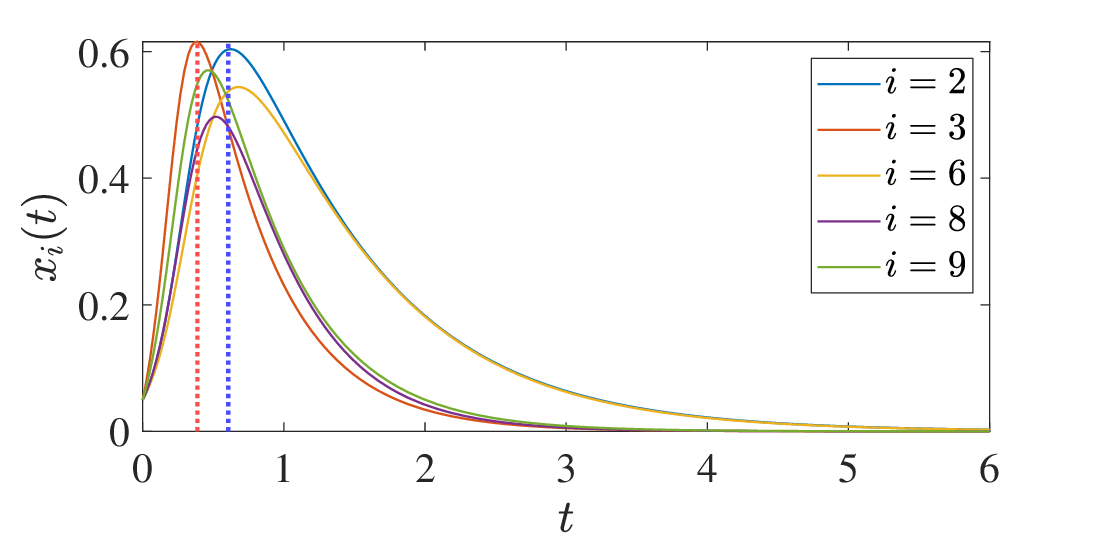}
    \end{subfigure}
    \begin{subfigure}{0.5\textwidth}
        \centering
        \includegraphics[width=0.75\linewidth]{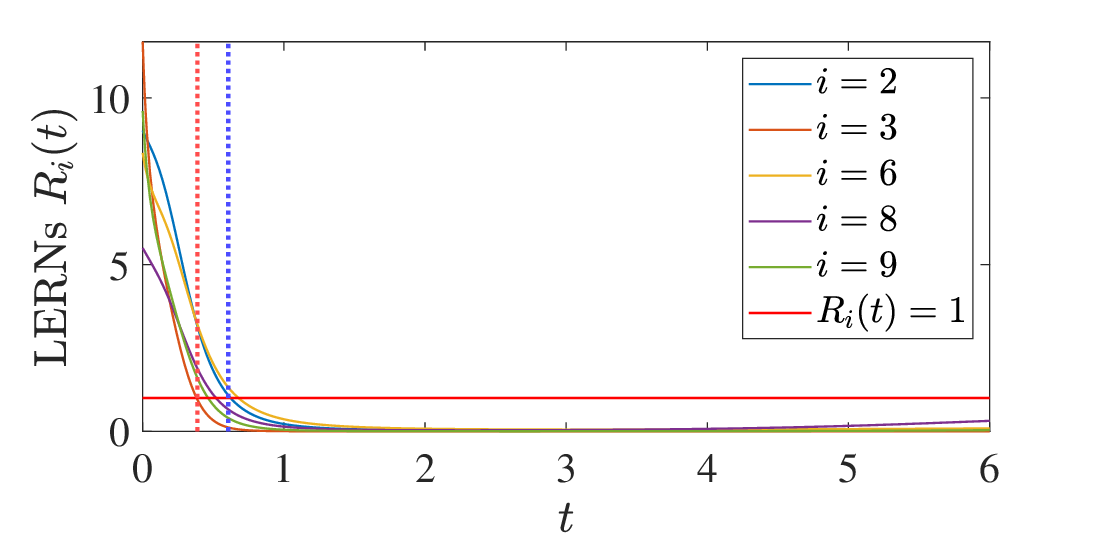}
    \end{subfigure}
    \caption{Evolution of $x_i(t)$ (top) and $R_i(t)$ (bottom) for $i\in\{2,3,6,8,9\}
    $ of the population network. Note that $R_i(t)$ can be non-monotonic and only crosses one once. Moreover, the peak infection time $\tau_{p_i}$ satisfies $R_i(\tau_{p_i})=1$ for all $i\in\mathcal{V^P}$.}
    \label{fig:node-level-pop}  
\end{figure}

In Fig.~\ref{fig:node-level-infra}, we show the evolution of the contamination of the resources nodes and the LERNs $R_j(t)$ for $j\in\{2,3,4\}$ of the infrastructure network. Consistent with Theorem~\ref{theo:Rt_local}, the contamination level $w_j(t)$ increases when $R_j(t)>1$ and decreases otherwise. However, note that for resource node~$2$, $R_2(t)$ crosses one multiple times. Thus, $R_j(t)=1$ does not necessarily coincide with the node-level peak infection time in the infrastructure network. 

\begin{figure}[h!]
    \centering
    \begin{subfigure}{0.5\textwidth}
        \centering
        \includegraphics[width=0.75\linewidth]{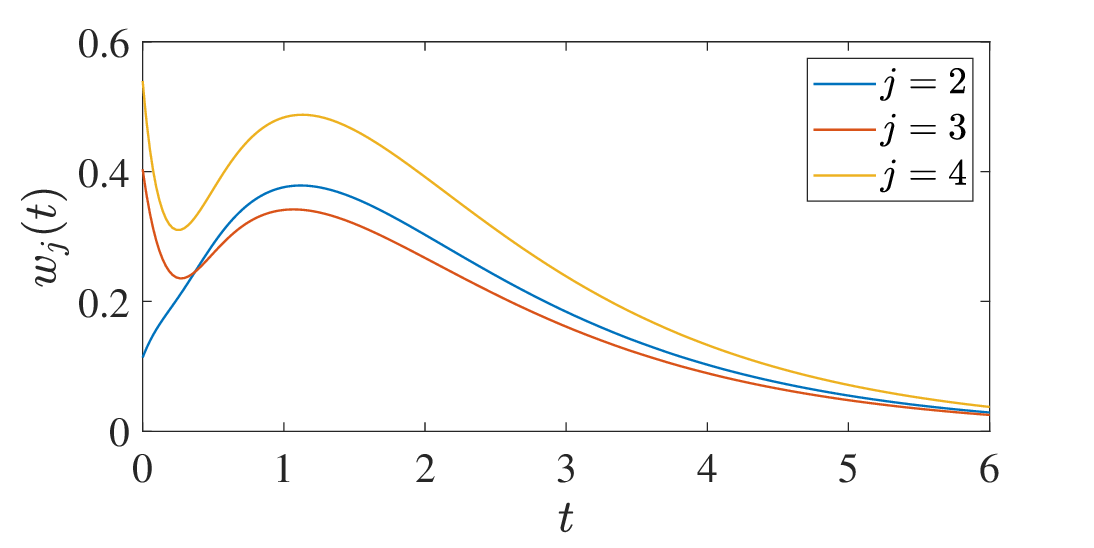}
    \end{subfigure}
    \begin{subfigure}{0.5\textwidth}
        \centering
        \includegraphics[width=0.75\linewidth]{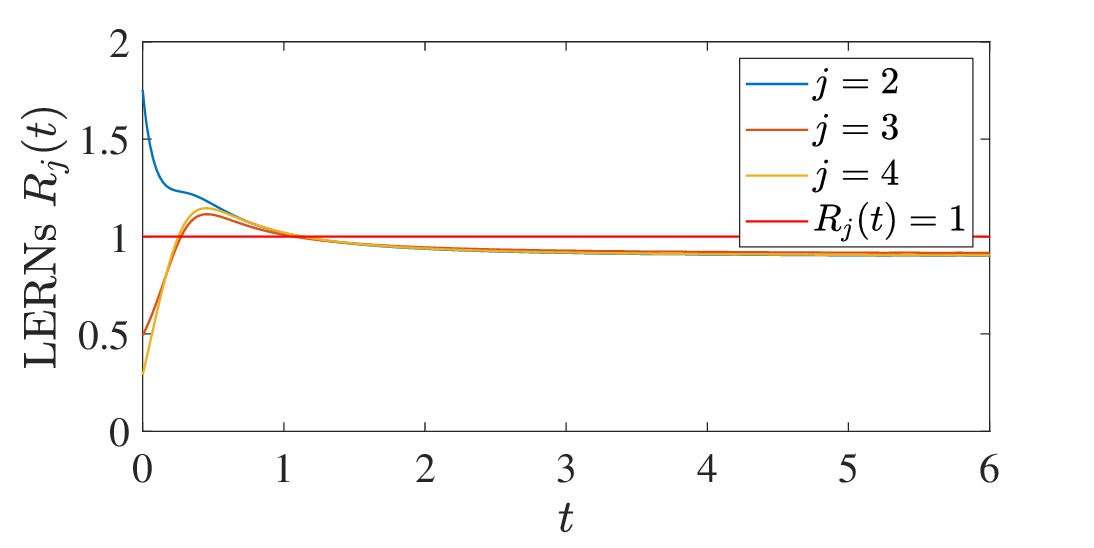}
    \end{subfigure}
    \caption{Evolution of the contamination level $w_j(t)$ (top) and the LERNs $R_j(t)$ (bottom) for $j\in\{2,3,4\}$ in the infrastructure network. All the claims in Theorem~\ref{theo:Rt_local} hold. However, $R_j(t)$ can cross one more than once.}
    \label{fig:node-level-infra} 
\end{figure}

The main takeaway from inspecting Figs.~\ref{fig:network-level},~\ref{fig:node-level-pop}, and \ref{fig:node-level-infra} is that the global effective reproduction number $R(t)$ does not account for the local spreading behavior. Furthermore, the more interesting node-level behavior of the infrastructure network, illustrated in Fig.~\ref{fig:node-level-infra}, is lost in the network-level analysis. Therefore, it is critical to leverage the DRNs and LERNs in order 
to predict and control the local behavior. 

\begin{figure}
        \centering
        \includegraphics[width=\linewidth]{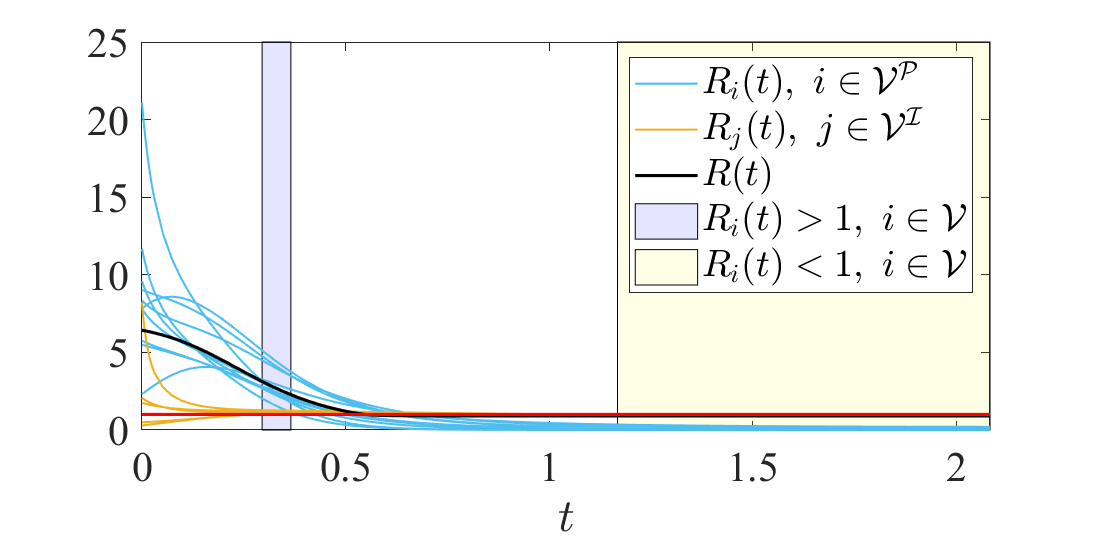}
    \caption{Evolution of the LERNs for all $i\in \mathcal{V}$. Theorem~\ref{theo:node_to_network}~\ref{theo:all_great_1} is depicted in the blue region. Theorem~\ref{theo:node_to_network}~\ref{theo:all_less_1} is depicted in the yellow region. 
    }
    \label{fig:thm3}  
\end{figure}

In Fig.~\ref{fig:thm3}, we illustrate the results in Theorem~\ref{theo:node_to_network}: the blue shaded area depicts where Theorem~\ref{theo:node_to_network}~\ref{theo:all_great_1} holds and the yellow shaded area depicts where Theorem~\ref{theo:node_to_network}~\ref{theo:all_less_1} holds. The main insight from inspecting Fig.~\ref{fig:thm3} is that leveraging the LERNs to predict the global behavior provides different (arguably better) insights than $R(t)$,
given that the $R(t)$ crosses one much earlier than when all the LERNs are below one, i.e., the infection is still increasing in more than half of the nodes.

\section{Conclusions}\label{sec:conclusions}

In this work, we introduce a novel SIR model that couples the dynamics of a virus spreading in a population network with the dynamics of contamination in an infrastructure network. We analyze the network-wide behavior of the system by introducing and leveraging the global effective reproduction number.
We also introduce the distributed and local effective reproduction numbers for both the population and infrastructure nodes in the multilayer networked model. We provide sufficient conditions to predict the monotonicity of the epidemic spreading at the node level using the distributed and local reproduction numbers. 
We explain how the node-level reproduction numbers can be used to analyze the transient behavior of the overall networked spreading process.  We illustrate our analytical results via simulations.

For future work, the connection between the node-level peak infection time and the local effective reproduction number equaling one should be shown analytically. Leveraging the node-level reproduction numbers and their connection to the global behavior for distributed control of the system is an interesting future direction. Finally, in this work, we assume that each node has knowledge of its state, its neighbors' states, and all of its local parameters; learning these parameters and estimating these states are vital steps in order to be able to leverage these tools in application.

\bibliographystyle{IEEEtran}
\bibliography{ref}

\end{document}